\documentclass[twoside,leqno,twocolumn]{article}

\usepackage[letterpaper]{geometry}

\usepackage{ltexpprt}

\newif\ifapx
\apxtrue 

\newcommand{\ourmaintitle}{Estimating Conditional Mutual Information for Discrete-Continuous Mixtures using Multi-Dimensional Adaptive Histograms}

\newcommand{\oururl}{\url{https://github.com/ylincen/CMI-adaptive-hist.git}}

\usepackage{latexsym}
\usepackage[ruled, vlined, nofillcomment, linesnumbered, algo2e]{algorithm2e}
\usepackage[notheorem]{eda} 
\usepackage{bm} 
\usepackage{cancel} 
\usepackage{tikz}
\usepackage{amsfonts}
\usepackage{footmisc}

\usepackage[pdftitle={ourtitle}, hidelinks, hyperfootnotes = false]{hyperref}

\graphicspath{ {./figs/} }

\ifpdf
\usetikzlibrary{external}
\fi

\SetKwComment{tcpas}{\{}{\}}
\SetCommentSty{textnormal}
\SetArgSty{textnormal}
\SetKw{False}{false}
\SetKw{True}{true}
\SetKw{Null}{null}
\SetKwInOut{Output}{output}
\SetKwInOut{Input}{input}
\SetKw{AND}{and}
\SetKw{OR}{or}
\SetKw{Continue}{continue}

\newcommand{\Kinit}{\ensuremath{K_{\text{init}}}\xspace}
\newcommand{\Kmax}{\ensuremath{K_{\text{max}}}\xspace}
\newcommand{\Score}{\ensuremath{\mathit{SC}}\xspace}
\DeclarePairedDelimiterX{\infdivx}[2]{(}{)}{%
  #1\;\delimsize\|\;#2%
}

\DeclareMathOperator*{\argmin}{\arg\!\min}
\DeclareMathOperator*{\argmax}{\arg\!\max}

\newcommand{\Models}{\ensuremath{\mathcal{M}}\xspace}
\newcommand{\dX}{\ensuremath{\mathcal{X}}}
\newcommand{\dY}{\ensuremath{\mathcal{Y}}}
\newcommand{\dZ}{\ensuremath{\mathcal{Z}}}

\newcommand{\estmlp}{\ensuremath{\hat{\theta}}\xspace}

\newcommand{\Indep}{\mathop{\perp\!\!\!\perp}\nolimits}

\tikzset{node/.style={black, draw=black, circle, minimum size=0.6cm, scale=0.7}} 
\tikzset{snode/.style={black, draw=black, fill=lightgray, circle, minimum size=0.6cm, scale=0.7}} 
\tikzset{dummy/.style={black, draw=black, circle, minimum size=0.55cm, scale=0.7}} 
\tikzset{latent/.style={black, draw=black, fill=lightgray, circle, minimum size=0.55cm, scale=0.7}} 

\tikzset{causes/.style={->,very thick,  color=black}} 
\tikzset{causesxor/.style={->,very thick, dashed, color=black}} 
\tikzset{causesxoro/.style={o->,very thick, dashed, color=black}} 
\tikzset{connected/.style={o-o,very thick, color=black}} 
\tikzset{connectedd/.style={o-o,very thick, dashed,  color=black}} 
\tikzset{ocauses/.style={o->,very thick,  color=black}} 
\tikzset{confounder/.style={<->,very thick,  color=black}} 
\tikzset{confounderxor/.style={<->,very thick, dashed, color=black}} 
\tikzset{confounderl/.style={<->,very thick,  color=black, bend left=45}} 
\tikzset{confounderr/.style={<->,very thick,  color=black, bend right=45}} 

\tikzstyle{flatlabel}  = [above, font = \tiny, inner sep = 1pt, text = black]
\tikzstyle{flatlabelb}  = [below, font = \tiny, inner sep = 1pt, text = black]
\tikzstyle{slopelabel}  = [sloped, above, font = \tiny, inner sep = 1pt, text = black]
\tikzstyle{slopelabelb}  = [sloped, below, font = \tiny, inner sep = 1pt, text = black]

\begin{document}
\setlength{\pdfpagewidth}{8.5in}
\setlength{\pdfpageheight}{11in}

\title{\ourmaintitle}
\author{Alexander Marx \thanks{equal contribution (ordered alphabetically)} \thanks{CISPA Helmholtz Center for Information Security and Max Planck Institute for Informatics, Saarland University. amarx@mpi-inf.mpg.de}
\and Lincen Yang \footnotemark[1]  \thanks{Leiden Institute of Advanced Computer Science, Leiden University. $\{$l.yang,m.van.leeuwen$\}$@liacs.leidenuniv.nl} \and Matthijs van Leeuwen \footnotemark[3]}

\date{}

\maketitle
   
\begin{abstract}
\small\baselineskip=9pt Estimating conditional mutual information (CMI) is an essential yet challenging step in many machine learning and data mining tasks. Estimating CMI from data that contains both discrete and continuous variables, or even discrete-continuous mixture variables, is a particularly hard problem. In this paper, we show that CMI for such mixture variables, defined based on the Radon-Nikodym derivate, can be written as a sum of entropies, just like CMI for purely discrete or continuous data. Further, we show that CMI can be consistently estimated for discrete-continuous mixture variables by learning an adaptive histogram model. In practice, we estimate such a model by iteratively discretizing the continuous data points in the mixture variables. To evaluate the performance of our estimator, we benchmark it against state-of-the-art CMI estimators as well as evaluate it in a causal discovery setting.

\end{abstract}

\section{Introduction}
\label{sec:intro}

In many research areas, such as classification~\cite{lee:13:classification}, feature selection~\cite{vinh:14:chisqcorrection}, and causal discovery~\cite{spirtes:00:book}, estimating the strength of a dependence plays a key role. A theoretically appealing way to measure dependencies is through mutual information (MI) since it has several important properties, such as the chain rule, the data processing inequality, and---last but not least---it is zero if (and only if) two random variables are independent of each other~\cite{cover:12:elements}. For structure identification, such as causal discovery, conditional mutual information (CMI) is even more interesting since it can help to distinguish between different graph structures. For instance, in a simple Markov chain $X \to Z \to Y$, $X$ and $Y$ may be dependent, but are rendered independent given $Z$. Vice versa, a collider structure such as $X \to Z \leftarrow Y$ may introduce a dependence between two marginally independent variables $X$ and $Y$ when conditioned on $Z$.

While estimating (conditional) mutual information for purely discrete or continuous data is a well-studied problem~\cite{cover:12:elements,darbellay:99:adaptive-partitioning,gao:16:kernelandnn,han:15:3h,paninski:08:kernel}, many real-world settings concern a mix of discrete and continuous random variables, such as age (in years) and height, or even random variables that can individually consist of a \emph{mixture} of discrete and continuous components. 
Although several discretization-based approaches that can estimate MI for a mix of discrete and continuous random variables have recently emerged~\cite{cabeli:20:mixedsc,mandros:20:fdmixed,suzuki:16:estimator}, so far only methods based on $k$-nearest neighbour ($k$NN) estimation were shown to work on \emph{mixed variables}, which may consist of discrete-continuous mixture variables \cite{gao:17:mixture,mesner:19:ms,rahimzamani:18:ravk}.

Regardless of the success of $k$NN-based estimators, discretization-based approaches have attractive properties, e.g., with regard to global interpretation. That is, a natural and understandable way to discretize a continuous random variable is via creating a histogram model, where we cut the sample space of the continuous variable in multiple non-overlapping parts called bins~\cite{scott2015multivariate}, or (hyper)rectangles for multi-dimensional variables. Within a bin, we consider the distribution to be constant, which allows us to estimate the density function via Riemann integration by making the bins smaller and smaller~\cite{cover:12:elements}. This definition, however, is less straightforward when mixed variables are involved.

In this paper, we approach the problem as follows: we first extend the definition of entropy for a univariate discrete-continuous mixture variable given by Politis~\cite{politis1991entropy} to multivariate variables. Using this definition, we show that CMI for mixed random variables can be written as a sum of entropies that are well-defined through the Radon-Nikodym derivate (see Sec.~\ref{sec:entropy-mixtures}). Exploiting this property, we propose a consistent CMI estimator for such data that is based on adaptive histogram models in Sec.~\ref{sec:histograms}. To efficiently learn adaptive histograms from data, in Sec.~\ref{sec:theory} we define a model selection criterion based on the minimum description length (MDL) principle~\cite{grunwald:07:book}. Subsequently, we propose an iterative greedy algorithm that aims to obtain the histogram model that minimizes the proposed MDL score in Sec.~\ref{sec:algo}. We discuss related work in Sec.~\ref{sec:related} and in Sec.~\ref{sec:exps}, we empirically show that our method performs favourably to state-of-the-art estimators for mixed data and can be used in a causal discovery setting.

\section{Entropy for Mixed Random Variables}
\label{sec:entropy-mixtures}

We consider multi-dimensional \emph{mixed random variables}, of which any individual dimension can be discrete, continuous, or a discrete-continuous mixture. Further, we call a vector of such mixed random variables a \emph{mixed random vector}. 
For a mixed random vector $(X,Y)$, where $X$ and $Y$ are possibly multivariate, we need to adopt the most general definition of mutual information (MI), i.e., the measure-theoretic definition:
\[
	I(X;Y) = \int_{\dX \times \dY} \log \frac{dP_{XY}}{dP_{X} P_{Y}} dP_{XY} \; ,
\]
where $dP_{XY}/(dP_{X}P_{Y})$ is the Radon-Nikodym derivative, $dP_{XY}$ the joint measure, and $P_{X}P_{Y}$ the product measure. It has been proven that $P_{X}P_{Y}$ is \emph{absolutely continuous} with respect to $P_{XY}$ \cite{gao:17:mixture}, 
i.e., $P_{XY} = 0$ whenever $P_{X}P_{Y} = 0$; and therefore, such a Radon-Nikodym derivative always exists and $I(X,Y)$ is well-defined. This measure-theoretic definition can be extended to CMI using the chain rule: $I(X;Y|Z) = I(X; \{ Y,Z \} ) - I(X;Z)$.

It is common knowledge that for purely discrete or continuous random variables, CMI can be written as a sum of entropies, i.e., $I(X;Y|Z) = H(X,Z) + H(Y,Z) - H(X,Y,Z) - H(Z)$. 
What is not clear, however, is if this formula also holds when $(X,Y,Z)$ contains discrete-continuous mixture random variables. 
We investigate this problem in two steps. We first define the measure-theoretic entropy for a (possibly multi-dimensional) discrete-continuous mixture random variable and prove it to be well-defined, though previous work claimed the opposite \cite{gao:17:mixture}. Second, using this definition, we prove that (conditional) MI for a mixed random vector can be written as the sum of measure-theoretic entropies, just like purely continuous or discrete random vectors.

\subsection{A Generalized Definition of Entropy} \label{subsec:defi_entropy}
The measure-theoretic entropy is defined only for one-dimensional random variables~\cite{politis1991entropy}. Building upon this definition, we give an explicit proof that such a one-dimensional measure-theoretic entropy is well-defined, and then extend this definition to the multi-dimensional case, which we prove is also well-defined. 

\subsubsection{Generalized One-Dimensional Entropy} \label{subsec:def_general_entropy_1d}
We start off by reviewing the existing definition for the one-dimensional case~\cite{politis1991entropy}. Given a one-dimensional random variable $X$, 
entropy $H$ is defined as 
\begin{equation} \label{eq:H1d}
	H(X) =  \int_{\mathbb{R}} \frac{dP_X(x)}{dv(x)} \log \frac{dP_X(x)}{dv(x)} dv(x),
\end{equation}
where $v(\cdot)$ is a measure defined on all one-dimensional Borel sets \cite{politis1991entropy}. If $v(\cdot)$ is the Lebesgue measure, which we denote as $u(\cdot)$, $H(X)$ becomes the differential entropy. Alternatively, if $v(\cdot)$ is a counting measure, $H(X)$ becomes the common (discrete) entropy. 

If, however, $X$ is a discrete-continuous mixture variable, $v$ is defined as follows. We split $\mathbb{R}$ into three disjoint subsets s.t. $\mathbb{R} = S_d \cup S_c \cup S_o$. First, $S_o$ is the subset of $\mathbb{R}$ on which $X$ has zero probability measure, i.e., $P_X(S_o) = 0$. Second, the set $S_d$ contains all discrete points, i.e., $S_d$ is countable and $\forall x \in S_d, P_X(x) >0$. Third, $S_c$ covers the continuous points, hence $P_X(S_c) + P_X(S_d) = 1$ and for any Borel set $A \subseteq S_c$ satisfying $u(A) = 0$, we have $P_X(A) = 0$. 
Based on these three subsets $S_d, S_c,$ and $S_o$, we can define $v$ as 
\begin{equation}
\label{eq:measurev}
v(A) = u(A \cap S_c) + |A \cap S_d| \; ,
\end{equation}
where $| A \cap S_d |$ is the cardinality of this intersection. 

To show that the generalized one-dimensional entropy is well-defined, we need to prove that the Radon-Nikodym derivative $dP_X/dv$ always exists. This we show in the following lemma.

\begin{lemma} \label{abs_cont_oned}
	Given a one-dimensional discrete-continuous random variable $X$ with probability measure $P_X$, $P_X$ is absolutely continuous w.r.t.\ $v$, i.e., $P_X = 0$ whenever $v=0$, and hence $dP_X/dv$ always exists.
\end{lemma}

We provide the proof of Lemma~\ref{abs_cont_oned}, as well as for Lemmas~\ref{lemma:abs_cont} and~\ref{lemma:3H} in Supplementary Material~\ref{sec:apx-proofs}.

\subsubsection{Generalized Multi-Dimensional Entropy} \label{subsec:def_multi_general_entropy}
In the following, we extend the measure-theoretic entropy definition to a mixed $k$-dimensional random vector $W = ( W_1, \dots, W_k )$.
For each $W_i$, we define $S_d^i, S_c^i, S_o^i$ and measure $v^i$ as above, and also define the \emph{product measure} $v$ for the $k$-dimensional random vector as 
$v = v^1 \times \ldots \times v^k$. Then, define the entropy for $W$ as
\begin{equation} \label{eq:Hmd}
	H(W) = \int_{\mathbb{R}^k} \frac{dP_W(w)}{dv(w)} \log \frac{dP_W(w)}{dv(w)} dv(w).
\end{equation}
To prove that such entropy is well-defined, we show that $dP_W/dv$ always exists. 
\begin{lemma} \label{lemma:abs_cont}
	Given a mixed $k$-dimensional random vector $W = ( W_1, \dots, W_k )$ with probability measure $P_W$, 
	$dP_W/dv$ always exists.
\end{lemma}
Last, based on Lemma~\ref{abs_cont_oned} and \ref{lemma:abs_cont}, we can prove that just like for a purely continuous or discrete random vector, conditional mutual information for a mixed random vector can be written as a sum of entropies. 
\begin{lemma} \label{lemma:3H}
	Given a mixed random vector $(X,Y,Z)$ with joint probability measure $P_{XYZ}$, we can write $I(X;Y|Z) = H(X,Z) + H(Y,Z) - H(Z) - H(X,Y,Z)$,
	 where each entropy can be defined as in Eq.~\eqref{eq:Hmd}.
\end{lemma}

As a direct implication of the above proof, it follows that mutual information can also be written as the sum of entropies, since it is a special case of CMI with $Z=\emptyset$. With this generalized definition, we can now show how to estimate CMI using adaptive histogram models.

\section{Adaptive Histogram Models}
\label{sec:histograms}
Adaptive histogram models have been thoroughly studied for continuous random variables~\cite{scott2015multivariate}; however, to the best of our knowledge, there exists no rigorous definition of histograms for mixed random variables. 
Thus, to use histogram models as a foundation to estimate the measure-theoretic (conditional) MI, we need to rigorously define histograms for mixed random variables.
We start with the one-dimensional case.

\subsection{One-Dimensional Histogram Models} \label{subsec:onedhist_defi}
A histogram model is typically defined based on a set of consecutive intervals called \emph{bins}~\cite{scott2015multivariate}. However, to deal with discrete-continuous mixture random variables, we define the set of bins, denoted as $B$, such that each bin is either an interval or a set containing only a single point. That is, $B = B' \cup B''$, where $B'$ and $B''$ are sets of subsets of $\mathbb{R}$, with $B'$ consisting of countable consecutive intervals and $B''$ consisting of countable single point sets. Last, we define the ``width" of a bin using the measure $v$ as defined in Eq.~\refeq{eq:measurev}, i.e., for a bin $B_j \in B$ we have
\begin{equation}
\label{eq:prod-measure-bins}
v(B_j) = u(B_j \cap B') + |B_j \cap B''| \; .
\end{equation}
As any $B_j \in B''$ contains only a single discrete point, $v(B_j) = 1$ for all $B_j \in B''$.

Further, we define a histogram model $M$ as a set of bins equipped with a parameter vector of length $K$, where $K = |B|$ is the number of bins. That is, a histogram model $M$ is a family of probability distributions $P_{X,\theta}$, parametrized by the vector $\theta = (\theta_1,\ldots,\theta_K)$. Each element of $\theta$ represents the Radon-Nikodym derivative (or density) of each bin. 
Note that this definition generalizes to purely continuous  random variables when $B'' = \emptyset$ and also to discrete random variables if $B' = \emptyset$. For the latter case, the histogram model degenerates to a multinomial model.

\subsection{Multi-Dimensional Histograms}\label{subsec:multihist_defi}
First, we define the set of multi-dimensional bins. For a mixed $k$-dimensional random vector $W = ( W_1,\ldots,W_k )$, 
we define the set of \emph{bins} for each $W_i$ as in Sec.~\ref{subsec:onedhist_defi}, denoted as $B^i$. Consequently, we can define a set of $k$-dimensional \emph{bins}, denoted $B$, by the Cartesian product $B = B^1 \times \ldots \times B^k$.

Since each $B^i$ is countable, $B$ is also countable, and we can hence assume $B$ is indexed by $j$. Then, we split $B$ in a similar way as in the one-dimensional case, i.e., $B = B'\cup B''$, where $B''$ contains \emph{only discrete values}. That is, for any $k$-dimensional bin $B_j \in B''$, 
each dimension of $B_j$ is a set that contains a single one-dimensional point. Note that, however, for any $B_j \in B'$, each dimension of $B_j$ can either be a one-dimensional interval or a one-dimensional single-point set.
Further, we define the volume of a multi-dimensional bin $B_j \in B$ using the product measure $v(B_j)$ (see Sec.~\ref{subsec:def_multi_general_entropy}).

Similar to one-dimensional histograms, a multi-dimensional histogram model $M$ can be described by a probability distribution $P_{W,\theta}$ parametrized by the vector $\theta = (\theta_1,\ldots,\theta_K)$, where $K$ is the number of bins and $\theta_i$ is the Radon-Nikodym derivative for each bin. 

\subsection{Maximum Likelihood Estimator}
Given a possibly multi-dimensional histogram with $K$ bins, we denote the Radon-Nikodym derivative $dP_{W,\theta}/dv$ as $f^h_{\theta}$
and its maximum likelihood estimator as $f^h_{\hat{\theta}}$. 
Observe that for any parameter $\theta_j \in \theta$, the product $\theta_j v(B_j)$ follows a multinomial distribution. Thus, given a dataset $D=\{D_i\}_{i=1,\ldots,n}$, with $D_i$ representing a row, the maximum log-likelihood is denoted as and equal to
\begin{equation} \label{eq:ml_hist}
       l_{M}(D) = \log f^h_{\estmlp(D)}(D) = \log \prod_{j = 1}^K \left( \frac{c_j}{n \cdot v(B_j)} \right)^{c_j},
\end{equation}
where $c_j$ and $v(B_j)$ are respectively the number of data points and the bin volumes of bin $j \in \{1 \ldots K\}$. Notice that this maximum likelihood generalizes to the purely discrete case (i.e., multinomial distribution) where all $v(B_j) = 1$, and to the purely continuous case~\cite{scott2015multivariate} where $v$ becomes the Lebesgue measure.

\subsection{Conditional Mutual Information Estimator}
Combining all previous theoretical discussions, we can now estimate conditional mutual information for three (possibly multivariate) random variables $X,Y$ and $Z$ by 
\begin{small}
\[
 I^h(X;Y|Z)=H^h(X,Z){+}H^h(Y,Z){-}H^h(X,Y,Z){-}H^h(Z) \; .
\]
\end{small}
The corresponding measure-theoretic entropies are estimated from $k$-dimensional data over $( X,Y,Z )$, where $k_X$, $k_Y$ and $k_Z$ are the corresponding number of dimensions of $X,Y$ and $Z$. We estimate the entropies as 
\begin{equation}
\begin{split}
	H^h(X,Y,Z) &= -\int_{\mathbb{R}^{k}} f^h_{\hat{\theta}}(x,y,z) \log (f^h_{\hat{\theta}}(x,y,z)) dv \\
	H^h(X,Z) &= -\int_{\mathbb{R}^{k_X+k_Z}} f^h_{\hat{\theta}}(x,z) \log (f^h_{\hat{\theta}}(x,z)) dv \\
	H^h(Y,Z) &= -\int_{\mathbb{R}^{k_Y + k_Z}} f^h_{\hat{\theta}}(y,z) \log (f^h_{\hat{\theta}}(y,z)) dv \\
	H^h(Z) &= -\int_{\mathbb{R}^{k_Z}} f^h_{\hat{\theta}}(z) \log (f^h_{\hat{\theta}}(z)) dv
\end{split}
\end{equation}
in which $f^h_{\hat{\theta}}(x,y,z)$ is the maximum likelihood estimator given the data, while we obtain $f^h_{\hat{\theta}}(x,z)$, $f^h_{\hat{\theta}}(y,z)$, and $f^h_{\hat{\theta}}(z)$ via marginalization from $f^h_{\hat{\theta}}(x,y,z)$.  Next, we will prove that $I^h$ is a strongly consistent estimator for conditional mutual information on mixed data.

\begin{theorem}
\label{th:convergence}
Given a mixed random vector $(X,Y,Z)$ with probability measure $P_{XYZ}$, 
\[
\lim_{v' \rightarrow 0} \lim_{n \rightarrow \infty} I^h(X;Y|Z) = I(X;Y|Z)
\]
almost surely, where $n$ refers to the sample size and $v'$ refers to the maximum of the histogram volumes for bins in $B'$ (defined in Sec.~\ref{subsec:multihist_defi}). 
\end{theorem}

The proof is provided in Supplementary Material~\ref{sec:apx-proofs}. Informally, our proof is based on the following key aspects: 1) All volume-related terms in $I^h$ cancel out, 2) discrete empirical entropy converges to the true entropy almost surely \cite{antos2001convergence}, and 3) 
in the limit, differential entropy can be obtained by discretizing a continuous random variable into ``infinitely" small bins \cite[Theorem 8.3.1]{cover:12:elements}. Notably, the order of the double limit in Theorem~\ref{th:convergence} inherently indicates that $n$ should grow faster than the number of bins~\cite{rudin1964principles}, which is also required for histograms on purely continuous data to converge~\cite{scott2015multivariate}.

\section{Learning Adaptive Histograms from Data}
\label{sec:theory}

To efficiently estimate a histogram model that inherits the consistency guarantees from Theorem~\ref{th:convergence} we need to consider the following requirements. First of all, we need to ensure that we learn a joint histogram model over $(X,Y,Z)$. This is due to the fact that we obtain the lower-dimensional entropies such as $H^h(X,Z)$ by marginalization over the likelihood estimator $f^h_{\hat{\theta}}(x,y,z)$. If we would not learn a joint model, the volume-related terms in $H^h(X,Y,Z), H^h(X,Z), H^h(Y,Z)$, and $H^h(Z)$ would not cancel out. In addition, we need to make sure that the number of bins is in $o(n)$ and increases if we were to increase the number of samples n, while at the same time the size of the bins decreases.

One way to achieve those properties would be to fix the bin width or the number of bins depending on the number of samples. However, such an approach is not very flexible and does not allow for variable bin widths. To allow for a more flexible model, we formally consider the problem of constructing an adaptive multi-dimensional histogram as a model selection problem and employ a selection criterion based on the minimum description length (MDL) principle \cite{rissanen:78:mdl}. MDL-based model selection has been successfully used for learning one-dimensional~\cite{kontkanen:07:histo} and two-dimensional histograms~\cite{kameya2011time,yang:20:unsupervised}, demonstrating adaptivity to both local density changes and sample size. 

We now briefly introduce the MDL principle and define the MDL-optimal histogram model. Specifically, while previous work~\cite{kameya2011time,kontkanen:07:histo,yang:20:unsupervised} only considers purely continuous data (or more precisely, data with arbitrarily small precision), we apply the MDL principle to mixed-type data, based on our rigorous definition of histogram models for mixed random variables. On top of that, we empirically show that our score fulfils the desired properties---i.e. the number of bins grows as $o(n)$.

\subsection{MDL and Stochastic Complexity}
\label{sec:mdlintro}
The minimum description length principle is arguably one of the best off-the-shelf model selection criteria~\cite{grunwald:07:book}, which has been successfully applied to many machine learning and data mining tasks. The general idea is to assign a code length to data $D$ compressed by a model $M$, e.g., a histogram model. Given a collection of candidate models, denoted as $\Models$, MDL selects the model $M^*$ that minimizes the joint code length of the model and the data.
Formally, our goal is to find
\begin{equation}
	M^* = \argmin_{M \in \Models} L(D|M) + L(M),
\end{equation}
where $L(D|M)$ denotes the code length\footnote{The code length $L$ denotes the number of bits needed to describe the given object. Hence, all logarithms are to base $2$ and $0 \log 0 = 0$.} of the data given the model, while $L(M)$ denotes the code length needed to encode the model.

The optimal way of encoding data $D$ given $M$, in the sense that it will result in minimax \emph{regret}, is to use the \emph{normalized maximum likelihood (NML)} code~\cite{grunwald:07:book}. Accordingly, the code length of the data is called \emph{stochastic complexity (SC)}, which is defined as the sum of the negative log-likelihood $-l_{M}(D)$, defined in Eq.~\refeq{eq:ml_hist}, and the \emph{parametric complexity} (also called \emph{regret}) $\log R(n,K)$~\cite{grunwald:07:book}. The parametric complexity of a histogram model with $K$ bins is given by~\cite{kontkanen:07:histo, yang:20:unsupervised}
\[
R(n,K) = \sum_{c_1 + \dots + c_K = n} \frac{n!}{c_1! \cdots c_K!} \prod_{i=1}^K \left(\frac{c_i}{n}\right)^{c_i} \; ,
\] 
and can be computed in sub-linear time~\cite{mononen:08:sub-lin-stoch-comp}.

\subsection{Code Length of the Model}
Given a dataset $D$ with $n$ rows and $k$ individual columns $D^j$, we now define the model class $\Models$.
First, we create fixed bins according to $B''$ (as defined in Sec.~\ref{subsec:multihist_defi}) per discrete value that occurs in $D_j$. 
Next, we enumerate all possible bins for $B'$ with fixed precision $\epsilon$. To this end, denote the remaining non-discrete data points in $D_j$ as $D^c_j$. If $D^c_j$ is empty $D_j$ corresponds to a discrete variable and we can stop here. Otherwise, we create all possible cut points for $D^c_j$ as
$C_j^{0} = \{\min(D^c_j), \min(D^c_j) + \epsilon, \ldots, \max(D^c_j)\}$.
By selecting a subset of cut points $C_j \subseteq C_j^0$, we get a valid solution for $B'$. We can enumerate all possible segmentations by enumerating each $C_j \subseteq C_j^0$.

By repeating this process for each dimension, we obtain our model class $\Models$. 
Further, we get the code length for a model $M \in \Models$ by encoding all combinations of cut points for each dimension~\cite{kontkanen:07:histo}, i.e., 
\begin{equation}
	L(M) = \sum_{j \in \{1,\ldots,k\}} L(C_j) = \sum_{j \in \{1,\ldots,k\}} \log {|C_j^0| \choose |C_j|} \; .
\end{equation}
This completes the definition of our final optimization score $L(D|M) + L(M)$. 

To proof consistency for this score, we need to show that the number of selected bins grows at rate $o(n)$. Since the theoretical analysis is rather difficult, we instead empirically demonstrate this property for Gaussian distributed data in Supplementary Material~\ref{apx:data-generation}. In the next section, we present an iterative greedy algorithm that optimizes our MDL score.

\section{Implementation}
\label{sec:algo}

In this section, we describe our algorithm to estimate the joint entropy $H(X_1, \dots, X_k)$ for a $k$-dimensional discrete-continuous mixture random vector. 

\subsection{Algorithm}
To discretize a one-dimensional random variable $X$, we first create bins for the discrete values of $X$ and then discretize the continuous values. 
We detect discrete data points by checking if a single value $x$ in the domain $\dX$ of $X$ occurs multiple times. If a user-defined threshold $t$, e.g., $5$ is reached, we create a bin for this point. 
To discretize the remaining continuous values, we start by splitting $\dX$ into $\Kinit$ equi-width bins, which we can safely choose from the complexity class $o(\sqrt{n})$ (see Supplementary Material~\ref{apx:data-generation}). 
Using dynamic programming, we compute the variable-width histogram model $M$ that minimizes $L(D,M)$ in quadratic time w.r.t. $\Kinit$~\cite{kontkanen:07:histo}.

Since the runtime complexity to compute the optimal variable-width histogram over a multi-dimensional random variable would grow exponentially w.r.t. $k$, we opt for an iterative greedy algorithm (we provide the pseudocode in Supplementary Material~\ref{apx:algorithm}). 
We start by initializing the optimization: for every dimension, we fix bins for the discrete values and put the remaining continuous values into a single bin. 
Then, in each iteration, we compute a candidate discretization for each dimension and keep the discretization of that dimension that provides the highest gain in compression. To compute a candidate discretization for a dimension $X_j$, we extend the one-dimensional algorithm described above. That is, we determine those cut points for $X_j$ that provide the highest gain in $L(D,M)$, while keeping the bins for the remaining dimensions fixed. 
We repeat this until the maximum number of iterations $i_{\text{max}}$ is reached or we cannot further decrease $L(D,M)$.

\subsection{Complexity}

The complexity of discretizing a univariate random variable is in $\mathcal{O}(\Kmax \cdot (\Kinit)^2)$ and depends on the number of initial bins $\Kinit$ and the maximum number of bins $\Kmax$, which we typically chose as a fraction of $\Kinit$ (both in $o(\sqrt{n})$). 
In a multi-dimensional setting we have to multiply this complexity by the current domain size of the remaining variables, since we have to update each bin conditioned on those. In the worst case, this number is equal to $(\Kmax)^{k-1}$. Overall, we apply this procedure---if all variables are continuous---$i_{\text{max}} \cdot k$ times.

\section{Related Work}\label{sec:related}

We discuss related methods for adaptive histograms and (conditional) mutual information estimation.

Both theoretical properties and practical issues of density estimation using histograms have been studied for decades \cite{scott2015multivariate}. 
Various algorithms have been proposed for the challenging task of constructing an adaptive one-dimensional histogram, among which the MDL-based histogram \cite{kontkanen:07:histo} is considered to be the state-of-the-art, as it is self-adaptive to both local density structure and sample size and does not have any hyperparameters.
Learning adaptive multivariate histograms is even harder due to the combinatorial explosion of the search space. One approach is to resort to the dyadic CART algorithm \cite{klemela2009multivariate}; various methods designed for specific tasks also exist \cite{kameya2011time,weiler2007multi}. Our algorithm is similar to that of Kameya~\cite{kameya2011time}, but they only consider the two-dimensional case. 

For discrete data, (conditional) mutual information estimation is a well-studied problem~\cite{cover:12:elements,han:15:3h,marx:19:sci,paninski:03:suboptimal,valiant:11:sub-linear} and it has been shown that mutual information can be estimated using the $3H$ principle~\cite{han:15:3h}. An important observation is that for discrete data, the empirical estimator for entropy is sub-optimal~\cite{paninski:03:suboptimal}, which encouraged the design of more efficient entropy estimators with sub-linear sample complexity~\cite{han:15:3h,valiant:11:sub-linear}.

For estimating (conditional) mutual information on continuous data or a mix of discrete and continuous data, three classes of approaches exist. The first class concerns kernel density estimation (KDE) methods~\cite{gao:16:kernelandnn,paninski:08:kernel}, which perform well on continuous data; however, no KDE-based MI and CMI estimation methods exist that are designed for discrete-continuous mixture random variables. Moreover, bandwidth tuning for KDE can be computationally expensive, which becomes even worse for mixed data, as different bandwidths may be needed for discrete random variables. The second class of methods relies on $k$-nearest neighbour ($k$NN) estimates~\cite{frenzel:07:fp,kozachenko:87:firstnn,kraskov:04:ksg}, which have been established as the state of the art~\cite{gao:17:mixture,rahimzamani:18:ravk}. $k$NN approaches can be applied not only to a mix of discrete and continuous variables, but can also be used as consistent MI~\cite{gao:17:mixture} and CMI~\cite{mesner:19:ms,rahimzamani:18:ravk} estimators for discrete-continuous mixtures.
The third class of methods first discretizes the continuous random variables and then calculates mutual information from the discretized variables~\cite{cover:12:elements,darbellay:99:adaptive-partitioning,suzuki:16:estimator}. Two recent approaches based on adaptive partitioning for mixed random variables have been proposed~\cite{cabeli:20:mixedsc,mandros:20:fdmixed}. While Mandros et al.~\cite{mandros:20:fdmixed} focus on mutual information and its application to functional dependency discovery, Cabeli et al.~\cite{cabeli:20:mixedsc}, similar to us, build upon an MDL-based score to estimate MI and CMI, to which we compare in Sec.~\ref{sec:exps}. The key difference is that Cabeli et al.~\cite{cabeli:20:mixedsc} compute $I(X;Y|Z)$ as $(I(X;\{Y,Z\}) - I(X;Z) + I(Y;\{X,Z\}) - I(Y;Z))/2$ and maximize each of the four terms (with penalty terms) directly, while we first learn a joint histogram. 

To the best of knowledge, we are the first to propose a CMI estimator for discrete-continuous mixture variables based on discretization or histogram density estimation. 
Our method can consistently estimate CMI on mixed random variables containing discrete-continuous mixtures. We focus on histogram-based models instead of $k$NN estimation, since histograms are more interpretable~\cite{scott2015multivariate} and do not require tuning of the parameter $k$, which can have a large impact on the outcome.

\section{Experiments}
\label{sec:exps}

In this section, we empirically evaluate the performance of our approach. First, we will benchmark our estimator against state-of-the-art CMI estimators on different data types. After that, we evaluate how well our estimator is suited to test for conditional independence in a causal discovery setup. For reproducibility, we make our code available online.\!\footnote{\oururl}\label{fn:appendix}

\newcommand{\RAVK}{\textsc{RAVK}\xspace}
\newcommand{\MS}{\textsc{MS}\xspace}
\newcommand{\FP}{\textsc{FP}\xspace}
\newcommand{\MIIC}{\textsc{MIIC}\xspace}
\newcommand{\correction}{\mathcal{C}\xspace}

\subsection{Mutual Information Estimation}

On the mutual information estimation task, we compare our approach to the state-of-the-art MI estimators. In particular, we compare against \FP~\cite{frenzel:07:fp}, \RAVK~\cite{rahimzamani:18:ravk} and \MS~\cite{mesner:19:ms}, which all rely on $k$NN estimates, and \MIIC~\cite{cabeli:20:mixedsc}, which is a discretization-based method. All of those can be applied to our setup, but only the authors of \RAVK and \MS specifically consider discrete-continuous mixture variables.
We apply \MIIC using the default parameters and use $k=10$ for all $k$NN-based approaches.\!\footnote{We evaluated all approaches with $k=5,10,20$. Since $k=10$ had the best trade-off and is close to $k=7$ as used by Mesner and Shalizi~\cite{mesner:19:ms}, we report those results.} For our algorithm, we set the maximum number of iterations and the threshold to detect discrete points in a mixture variable to $5$, set $\Kinit = 20 \log n$ and $\Kmax = 5 \log n$.
To comply with the literature, we compute all entropies in this section using the \emph{natural logarithm}. 

\subsubsection{Experiment I-IV}

As a sanity check, we start with an experiment on purely continuous data.
That is, for \textbf{Experiment I}, let $X$ and $Y$ be Gaussian distributed random variables with mean 0, variance 1, and covariance $0.6$. Consequently, the correlation $\rho$ between $X$ and $Y$ is $0.6$ and true MI can be calculated as $I(X;Y) = - \frac{1}{2} \log(1-\rho^2)$. In \textbf{Experiment II}, $X$ is discrete and drawn from $\text{Unif}(0,m-1)$, with $m=5$ and $Y$ is continuous with $Y \sim \text{Unif}(x,x+2)$ for $X=x$. Therefore, $I(X;Y) = \log(m) - \frac{(m-1) \log 2}{m}$~\cite{gao:17:mixture}. Next, for \textbf{Experiment III}, $X$ is exponentially distributed with rate $1$ and $Y$ is a zero-inflated Poissonization of $X$---i.e., $Y=0$ with probability $p = 0.15$ and $Y \sim \text{Pois}(x)$ for $X=x$ with probability $1-p$. The ground truth is $I(X;Y) = (1 -
p)(2 \log 2 - \gamma - \sum_{k=1}^{\infty} \log k \cdot 2^{-k} ) \approx (1-p)0.3012$, where $\gamma$ is the Euler-Mascheroni constant~\cite{gao:17:mixture}. Last, in \textbf{Experiment IV}, we generate the data according to the Markov chain $X \to Z \to Y$ (see Mesner and Shalizi~\cite{mesner:19:ms}). In particular, $X$ is exponentially distributed with rate $\frac{1}{2}$, $Z \sim \text{Pois}(x)$ for $X=x$ and $Y$ is binomial distributed with size $n=z$ for $Z=z$ and probability $p=\frac{1}{2}$. Due to the Markov chain structure, the ground truth $I(X;Y \mid Z) = 0$.

For each of the above experiments, we sample data with sample size $n \in \{ 100, 200, \dots, 1 \, 000 \}$ and generate $100$ data sets per sample size. We run each of the estimators on the generated data and show the mean squared error (MSE) of each estimator in Fig.~\ref{fig:ground-truth}.
Overall, our estimator performs best or very close to the best throughout the experiments and reaches an MSE lower than $0.001$ with at most $1 \, 000$ samples. The best competitors are \MS and \MIIC; however, both are biased when we consider discrete-continuous mixture variables, as we show in Experiment V. 

\begin{figure}[t]
	\begin{minipage}[t]{0.5\linewidth}
	\includegraphics[]{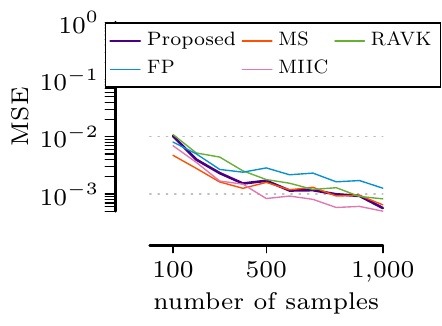}
	\end{minipage}%
	\begin{minipage}[t]{0.5\linewidth}
	\includegraphics[]{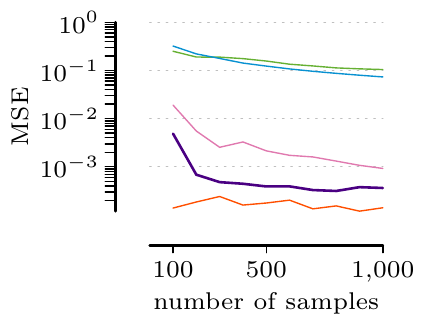}
	\end{minipage}%
	\linebreak
	\begin{minipage}[t]{0.5\linewidth}
	\includegraphics[]{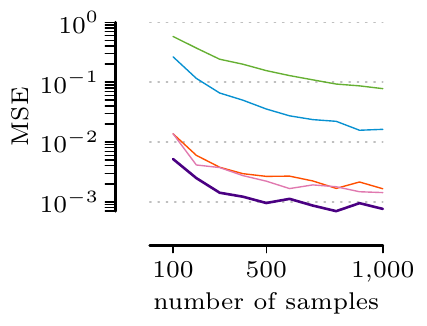}
	\end{minipage}%
	\begin{minipage}[t]{0.5\linewidth}
	\includegraphics[]{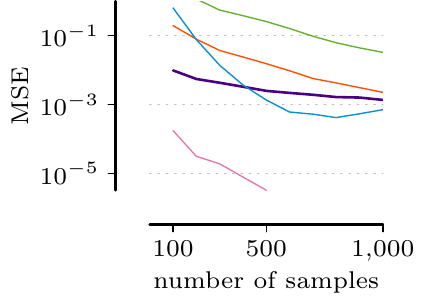}
	\end{minipage}%
	\caption{Synthetic data with known ground truth. Ordered from top-left to bottom-right, we show the MSE for Experiments I-IV, for our estimator and competing algorithms \MS, $\RAVK$, \FP and \MIIC.}
	\label{fig:ground-truth}
\end{figure}

\subsubsection{Experiment V}

Next, we generate data according to a discrete-continuous mixture~\cite{gao:17:mixture}. Half of the data points are continuous, with $X$ and $Y$ being standard Gaussian with correlation $\rho = 0.8$, while the other half follows a discrete distribution with $P(1, 1) = P(-1,-1) = 0.4$ and $P(1,-1) = P(-1,1) = 0.1$. In addition, we generate $Z$ independently with $Z \sim \text{Binomial(3, 0.2)}$. Hence the ground truth is equal to $I(X;Y) = I(X;Y \mid Z) = 0.4 \cdot \log \frac{0.4}{0.5^2} + 0.1 \cdot \log \frac{0.1}{0.5^2} - \frac{1}{4} \log (1 - 0.8^2) \approx 0.352$.

In Fig.~\ref{fig:ground-truth-md} (top) we show the mean and MSE for this experiment. We see that our estimator starts by overestimating the true value, but its average quickly converges to the true value, while the competing estimators seem to have a slightly positive or negative bias. Especially \FP and \MIIC, which were not designed for this setup, have a clear bias even for $1 \, 000$ data points.
The same trend can be observed for MSE.

\subsubsection{Experiment VI}

Last, we test how sensitive our method is to dimensionality. We generate $X$ and $Y$ as in Experiment II, but fix $n$ to $2 \, 000$ and add $k$ independent random variables, $Z_k \sim \text{Binomial(3, 0.5)}$. 

Fig.~\ref{fig:ground-truth-md} (bottom) shows the mean and MSE. Our estimator recovers the true CMI up to a negligible error up to $k=2$. After that, it starts to slowly underestimate the true CMI. This can be explained by the fact that the model costs increase linearly with the domain size and hence, we will fit fewer bins to the continuous variable for large $k$. We validated this conjecture by repeating the experiment for $n= 10 \, 000$. On this larger sample size, the MSE for our estimator remained below $0.001$ even for $k=4$. While \MIIC is slightly more stable for $k \ge 3$, the competing $k$NN-based estimators deviate quite a bit from the true estimate for higher dimensions.

\begin{figure}[t]
	\begin{minipage}[t]{0.5\linewidth}
	\includegraphics[]{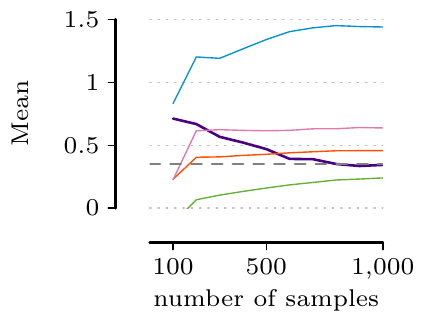}
	\end{minipage}%
	\begin{minipage}[t]{0.5\linewidth}
	\includegraphics[]{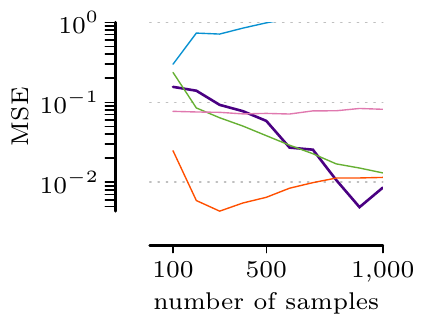}
	\end{minipage}%
	\linebreak
	\begin{minipage}[t]{0.5\linewidth}
	\includegraphics[]{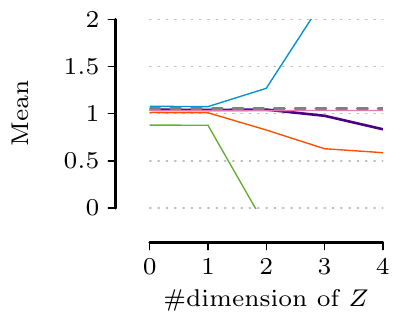}
	\end{minipage}%
	\begin{minipage}[t]{0.5\linewidth}
	\includegraphics[]{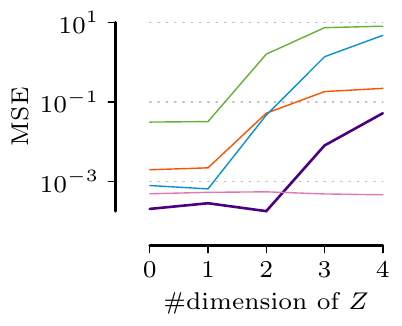}
	\end{minipage}%
	\caption{Top row: Experiment V, where we show the mean of the estimators (left) with the true CMI as a dashed gray line and the MSE (right). Bottom row: Experiment VI, where the sample size is constant at $2 \, 000$ and the x-axis refers to the number of dimensions of $Z$. We show the mean (left) and MSE (right). The color coding is chosen as in Fig.~\ref{fig:ground-truth}.}
	\label{fig:ground-truth-md}
\end{figure}

Overall, we are on par with or outperform the best competitor throughout Experiments~I--VI. Especially on mixture data, which is our main focus, our method is the only one that converges to the true estimate.

\subsection{Independence Testing}
\label{sec:indep-experiments}

In theory, two random variables $X$ and $Y$ are conditionally independent given a set of random variables $Z$, denoted as $X \Indep Y \mid Z$, if $I(X;Y \mid Z) = 0$. Vice versa, $X$ and $Y$ are dependent given $Z$, if $I(X;Y \mid Z) > 0$. In practice, we cannot simply rely on our estimator to conclude independence: due to the monotonicity of mutual information, i.e., $I(X;Y) \le I(X;Y \cup Z)$, estimates will rarely be \emph{exactly zero} in the limited sample regime, but only \emph{close to zero}~\cite{marx:19:sci,vinh:14:chisqcorrection}. To address this problem, we use our algorithm to discretize $X, Y$ and $Z$, and compute $I_{\correction}(X; Y | Z) := \max \{ 0, I_n(X_d; Y_d | Z_d) + \correction_n(X_d; Y_d | Z_d) \}$,
where $\correction_n$ is a correction term calculated from the discretized variables, which is negative. In the following, we evaluate our estimator with two different correction criteria. The first one is a correction for mutual information based on the Chi-squared distribution, with $\correction_n$ equal to $- \mathcal{X}_{\alpha, l} / 2n$~\cite{vinh:14:chisqcorrection}, where $\mathcal{X}_{\alpha, l}$ refers to the critical value of the Chi-squared distribution with significance level $\alpha$ and degrees of freedom $l$. We can compute the degrees of freedom $l$ from the domain sizes of the discretized variables for the conditional case as $l = (|\dX|-1)(|\dY|-1) |\dZ|$~\cite{suzuki:19:consistency}, and for the unconditional case as $l  =(|\dX|-1)(|\dY|-1)$. For the second correction, we replace each empirical entropy in $I_n(X_d; Y_d | Z_d)$ with its corresponding stochastic complexity term as defined in Sec.~\ref{sec:mdlintro}. If we subtract the regret terms for $H_n(X_d,Y_d,Z_d)$ and $H_n(Z_d)$ from those for $H_n(X_d,Z_d)$ and $H_n(Y_d,Z_d)$, we are guaranteed to get a negative value, thus a valid regret term~\cite{marx:19:sci}.
In the following, we refer to the test using the Chi-squared correction as $I_{\mathcal{X}^2}$ and to the one based on stochastic complexity as $I_{\text{SC}}$.

To test how well $I_{\mathcal{X}^2}$ and $I_{\text{SC}}$ perform on mixed-type and continuous data, we benchmark both against state-of-the-art kernel-based tests RCIT and RCoT~\cite{strobl:19:rcit}, as well as JIC~\cite{suzuki:16:estimator}, and \MIIC~\cite{cabeli:20:mixedsc}, which are both discretization-based methods using a correction based on stochastic complexity.\!\footnote{Note that MIIC calculates stochastic complexity based on factorized NML and JIC uses an asymptotic approximation of stochastic complexity, while we use quotient NML for $I_{\text{SC}}$~\cite{marx:19:sci,silander:18:qnml}.} To apply RCIT and RCoT on mixed data, we treat the discrete data points as integers. In the following, we evaluate the performance of each test in a causal discovery setup. In addition, we provide a more detailed description of the data generation and experiments on individual collider and non-collider structures in Supplementary Material~\ref{apx:data-generation}.

\subsubsection{Causal Discovery}

\begin{figure}[t]%
	\centering
	\includegraphics[]{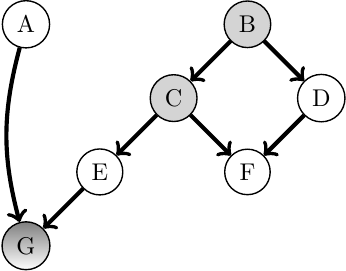}
	\caption{Synthetic network with continuous (white), discrete (gray) and mixed (shaded) random variables consisting of different causal structures, such as colliders, a chain ($C \to E \to G$), and a fork ($C \leftarrow B \to D$).}
	\label{fig:synthetic-network}
\end{figure}

\begin{figure}[t]
	\begin{minipage}[t]{0.5\linewidth}
	\includegraphics[]{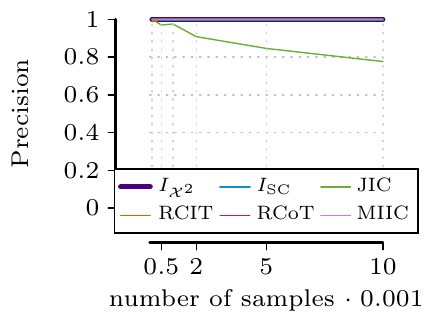}
	\end{minipage}%
	\begin{minipage}[t]{0.5\linewidth}
	\includegraphics[]{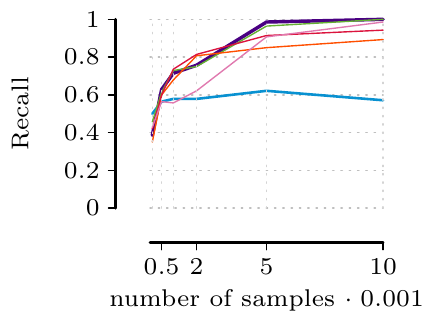}
	\end{minipage}%
	\caption{Precision (left) and recall (right) on undirected graphs inferred using the PC-stable algorithm equipped with the corresponding independence test. The data is generated from the graph shown in Fig.~\ref{fig:synthetic-network}.}
	\label{fig:prrec}  
\end{figure}

To evaluate our test in a causal discovery setting, we generate data according to a small synthetic network---shown in Fig.~\ref{fig:synthetic-network}---that consists of a mixture of generating mechanisms that we used in experiments I-IV and includes continuous and discrete (ordinal) random variables and one mixture variable, which is partially Gaussian and partially Poisson distributed (for details see Supplementary Material~\ref{apx:data-generation}).
To evaluate how well the ground truth graph can be recovered, we apply the PC-stable algorithm~\cite{colombo:12:pcmodification,spirtes:00:book} equipped with the different independence tests, where we use $\alpha = 0.01$ for $I_{\mathcal{X}^2}$, RCIT and RCoT. 

Fig~\ref{fig:prrec} shows recovery precision and recall for the undirected graph, averaged over $20$ draws per sample size $n \in \{ 100, 500, 1 \, 000,$ $2 \, 000, 5 \, 000, 10 \, 000 \}$.

We see that overall $I_{\mathcal{X}^2}$ performs best and is the only method that reaches both a perfect accuracy and recall. While JIC also reaches a perfect recall, it finds too many edges leading to a precision of only $80 \%$. Although also MIIC, RCIT and RCoT have a perfect precision, their recall is worse than for $I_{\mathcal{X}^2}$. Neither of the kernel-based tests manages to recall all the edges even for $10 \, 000$ samples. After a closer inspection, this is due to the edge $E \to G$ that involves the discrete-continuous variable $G$. If we compare $I_{\mathcal{X}^2}$ to $I_{\text{SC}}$, we clearly see that the latter is too conservative, which leads to a bad recall.

\section{Conclusion}\label{sec:conclusion}
We proposed a novel approach for the estimation of conditional mutual information from data that may contain  discrete, continuous, and mixture variables. To be able to deal with discrete-continuous mixture variables, we defined a class of generalized adaptive histogram models. Based on our observation that CMI for mixture-variables can be written as a sum of entropies, we presented a CMI estimator based on such histograms, for which we proved that it is consistent.

Further, we used the minimum description length principle to formally define optimal histograms, and proposed a greedy algorithm to practically learn good histograms from data. Finally, we demonstrated that our algorithm outperforms state-of-the-art (conditional) mutual information estimation methods, and that it can be successfully used as a conditional independence test in causal graph structure learning. Notably, for both setups, we observe that our approach performs especially well when mixture variables are present.

\section*{Acknowledgements}
This work is part of the research programme `Human-Guided Data Science by Interactive Model Selection' with project number 612.001.804, which is (partly) financed by the Dutch Research Council (NWO).

\newpage

\bibliographystyle{IEEEtranS}
\bibliography{bib/abbrev,bib/bib-jilles,bib/bib-paper}

\ifapx
\newpage
\appendix
\section*{Supplementary Material}
\setcounter{section}{19}

The supplementary material is structured as follows. First, we provide proofs for all lemmas and theorems. After that, we provide the pseudocode for our algorithm. Last, we provide additional experiments and details for the data generation for the causal discovery experiment.

\subsection{Proofs}
\label{sec:apx-proofs}

\subsubsection{Proof of Lemma~\ref{abs_cont_oned}}
\begin{proof}
	Given a Borel set $A \subseteq \mathbb{R}$ such that $v(A) = u(A \cap S_c) + |A \cap S_d| = 0$, we have $u(A \cap S_c) = 0$ due to non-negativity of any measure, as well as $|A \cap S_d| = 0$. Since $A \cap S_c \subseteq S_c$, by the definition of $S_c$ we have $P(A \cap S_c) = 0$.
	It remains to show that $A \cap S_d = \emptyset$, which we do by contradiction. Assume that $A \cap S_d \neq \emptyset$, then there exists $x \in A \cap S_d$ s.t. for a set containing only $x$, $|\{x\}| = 1$. Then $|A \cap S_d| \geq |\{x\}| = 1$, which contradicts $|A \cap S_d| = 0$. Thus, we must have $A \cap S_d = \emptyset$ and then $P_X(A) = 0$.
\end{proof}

\subsubsection{Proof of Lemma~\ref{lemma:abs_cont}}
\begin{proof}
	Given a $k$-dimensional Borel set $A$, there exist one-dimensional Borel sets $A_1,\ldots, A_k$ such that $A = A_1 \times \ldots \times A_k$. If $v(A)=0$, then there exists at least one $v^i, i \in \{1,\ldots,k\}$, such that $v^i(A_i)=0$. Thus, by Lemma~\ref{abs_cont_oned}, $P_{W_i}(A_i) = 0 \Rightarrow P_W(\mathbb{R}\times \ldots \times \mathbb{R} \times A_i \times \mathbb{R}\times \ldots \times \mathbb{R}) = 0 \Rightarrow P_{W}(A) = 0$, as $A = A_1 \times \ldots \times A_k \subseteq \mathbb{R}\times \ldots \times \mathbb{R} \times A_i \times \mathbb{R}\times \ldots \times \mathbb{R}$.
\end{proof}

\subsubsection{Proof of Lemma~\ref{lemma:3H}}
\begin{proof}
	We first proof the statement for $Z \neq \emptyset$, for which we can write $I(X;Y|Z) = I(X;\{Y,Z \}) -  I(X;Z)$ by the chain rule for mutual information. Thus, it suffices to prove that $I(X;Z) = H(X) + H(Z) - H(X,Z)$ and $I(X; \{Y,Z \}) = H(X) + H(Y,Z) - H(X,Y,Z)$. 
	Next, denote $v$ as the product measure defined based on $(X,Z)$, where $v = v^1 \times \ldots \times v^{k_{XZ}}$, and $k_{XZ}$ is the number of dimensions of $X$ plus that of $Z$; then by Lemma~\ref{lemma:abs_cont}, we also have $P_{XZ} \ll v$.
	Then, we show that $P_XP_{Z} \ll v$. For some $k_{XZ}$-dimensional Borel set $A = A_1 \times \ldots \times A_{k_{XZ}}$ satisfying $v(A) = 0$ there exists $v^i \in \{v^1,\ldots,v^{k_{XZ}}\}$ such that $v^i(A_i)=0$. Hence, $P_XP_{Z}(A) = 0$
	because $0 \leq P_XP_{Z}(A) = P_XP_{Z}(A_1 \times \ldots \times A_{k_{XZ}}) \leq P_XP_{Z}(\mathbb{R} \times \ldots \mathbb{R} \times A_i \times \mathbb{R} \ldots \times \mathbb{R}) = P_i(A_i) = 0$, where $P_i$ is the marginalization of the product measure $P_XP_Z$ to the $i$th dimension and $P_i(A_i) = 0$ is because $v^i(A_i)=0$ by the definition of $v$. 
	
	Finally, by the chain rule of the Radon-Nikodym derivative we have that 
	\begin{align}
	I(X;Z) &= \int \log \frac{dP_{X{Z}}}{dP_XP_{Z}} d P_{XZ} \\
	&= \int \log \frac{dP_{XZ}/dv}{dP_{X}P_{Z}/dv}(dP_{XZ}/dv)dv \\
	&= H(X) + H(Z) - H(X,Z) \; .
	\end{align} 
	The proof for $I(X;\{Y,Z\})$ is equivalent. If $Z = \emptyset$, CMI reduces to $I(X;Y)$, for which we can prove the statement in the same manner.
\end{proof}

\subsubsection{Proof of Theorem~\ref{th:convergence}}
To proof Theorem~\ref{th:convergence} we need several intermediate results. Lemma~\ref{le:histoDiscretization} shows that a histogram results in a valid discretization as all terms corresponding to volumes in $I^h$ cancel out, and hence $I^h$ can be written as a sum of plug-in estimators of \emph{discrete entropies}. Then, Lemma~\ref{le:discrete-convergence} shows a classic result that the plug-in estimator of discrete entropies will converge to the true entropy almost surely. Further, we show in Lemma~\ref{le:histoToTrue} that as the volumes of histogram bins containing continuous values go to 0, the true entropies of discretized variables (which are discretized by the histogram) converges to the true entropy of original variables. 

\begin{Definition}
	Given discrete random variables $X_d,Y_d,Z_d$ (possibly multi-dimensional), with support $S_d^X, S_d^Y,S_d^{Z}$, and given dataset $D = (x_i,y_i, z_i)_{i \in \{1,\ldots,n\}}$ with sample size $n$, the plug-in estimator of discrete entropy $H$ is denoted and defined as
	$$H_n(X_d,Y_d,Z_d) = -\sum_{j \in S_d^X \times S_d^Y \times S_d^{Z}} \hat{p}(j) \log \hat{p}(j)$$ 
	with probability estimates 
	$$\hat{p}(j) = \frac{|\{(x_i,y_i,z_i)_{i \in \{1,\ldots,n\}} : (x_i,y_i,z_i) = q_j \}|}{n} \; ,$$
	where $|\cdot|$ represents the cardinality of a set, and $q_j$ is the $j$th element in $S_d^X \times S_d^Y \times S_d^{Z}$.
\end{Definition}

\begin{lemma}
\label{le:discrete-convergence}
	Given a discrete random vector $(X_d,Y_d,Z_d)$, $lim_{n \rightarrow \infty} H_n(X_d,Y_d,Z_d) = H(X_d,Y_d,Z_d)$ almost surely~\cite{antos2001convergence}. 
\end{lemma}

\begin{lemma}
\label{le:histoToTrue}
	Given a random vector $(X,Y,Z)$ that contains discrete-continuous mixture random variables, with bins $B = B'\cup B''$ and the resulting discretized random vector $(X_d, Y_d, Z_d)$, where $B''$ contains discrete data points (of which every dimension has a discrete value) and $B' = B \setminus B''$, we have 
	\[
	\lim_{v' \rightarrow 0} H(X_d, Y_d, Z_d) = H(X, Y, Z) \; ,
	\]
	where $v' = \max_{B_j \in B'}(v(B_j))$.
\end{lemma}

\begin{proof}
	Firstly, it is well-known that this result holds if $(X,Y,Z)$ is a continuous random vector \cite{cover:12:elements}; then, if $(X,Y,Z)$ contains mixture variables, 
	\begin{align}
	H(X,Y,Z) &= \lim_{v' \rightarrow 0} \sum_{B_j \in B'} \frac{P_{X_dY_dZ_d}}{v(B_j)} \log \frac{P_{X_dY_dZ_d}}{v(B_j)}  \\
	&\; \; \; \; \; \; \: \, + \sum_{B_j \in B''} \frac{P_{X_dY_dZ_d}}{v(B_j)} \log \frac{P_{X_dY_dZ_d}}{v(B_j)} \\
	&= \lim_{v' \rightarrow 0} H(X_d,Y_d, Z_d) \; ,
	\end{align}
	which concludes the proof.
\end{proof}
\begin{Definition}
	Given a random vector $(X,Y,Z)$ that contains mixture variables, and an adaptive grid $B$, we define the discretized random variable $X_d, Y_d, Z_d$, with probability measure (probability mass function) 
	$$P_{X_d,Y_d,Z_d}(  (j_1,j_2,j_3) ) = \int_{B_j} \frac{d_{XYZ}}{dv}dv \; ,$$ 
	where $B_j$ denotes the $j$th bin of $B$. 
\end{Definition}

\begin{lemma}
\label{le:histoDiscretization}
	Given a $k$-dimensional random vector $(X,Y,Z)$ that contains mixture variables with an unknown probability measure $P_{XYZ}$, a dataset $D = (x_i,y_i,z_i)_{i \in \{1,\ldots,n\}}$ generated by $P_{XYZ}$, a histogram model $M$, and corresponding discretized random vector $(X_d,Y_d,Z_d)$, we have $I^h(X,Y|Z)$ is equal to 
	\[
	H_n(X_d, Z_d) + H_n(Y_d, Z_d) - H_n(X_d, Y_d, Z_d)- H_n(Z_d) \; .
	\]
	That is, the terms corresponding to volumes in $I^h$ cancel out and our histogram model results a valid discretization. 
\end{lemma}

\begin{proof}
Denote the adaptive grid of histogram model $M$ as $B^{XYZ}$, which is the Cartesian product of bins defined on $X,Y,Z$---i.e. $B^{XYZ} = B^X \times B^Y \times B^{Z}$, and denote the corresponding MLE of histogram density function as $f^h_{\hat{\theta}_{XYZ}}$. Further, define a function $v_X$, such that for each $x_i$ in $D$, $v_X(x_i) = v(B^X_j)$ if $x_i \in B^X_j$, where $B^X_j$ is a bin of $B^X$ and $v$ is defined based on the random variable $X$. Then, define $v_Y, v_{Z}, v_{XZ}, v_{YZ}, v_{XYZ}$ similarly. 

By the definition $I^h(X,Y|Z)$ is equal to
\[
H^h(X, Z) + H^h(Y, Z) - H^h(X, Y, Z)- H^h(Z) \; .
\]
First consider $H^h(X,Z)$. We write $B^{XZ} = B^X \times B^{Z}$, with marginal density function $f^h_{\hat{\theta}_{XZ}}$. W.l.o.g. suppose that $B_{XZ}$ consists of $K$ bins, denoted as $B^{XZ}_j, j \in \{1,\ldots,K\}$. Then,
\begin{equation}
	\begin{split}
		 H^h(X, Z)  &= -\int_{\mathbb{R}^{k_X+k_Z}} f^h_{\hat{\theta}_{XZ}} \log f^h_{\hat{\theta}_{XZ}} dv \\
		& = -\sum_{j=1}^{K} \int_{B^{XZ}_j} f^h_{\hat{\theta}_{XZ}} \log f^h_{\hat{\theta}_{XZ}} dv\\
		&= -\sum_{j=1}^{K} c_j \log\left(\frac{c_j}{n v(B_j)}\right) \\
		&= -\sum_{j=1}^{K} c_j \log\left(\frac{c_j}{n}\right) + \sum_{i=1}^n \log(v_{XZ}(x_i,z_i)) \\
		&= H_n(X_d, Z_d) + \sum_{i=1}^n \log(v_{XZ}(x_i,z_i)) \; ,
	\end{split}
\end{equation}
where $c_j$ is the number of data points in $B_j$ and $v_{XZ}(x_i, z_i) = v_X(x_i)v_Z(z_i)$. The remaining entropies can be calculated similarly. Hence, $I^h(X,Y|Z) =  H_n(X_d, Z_d) + H_n(Y_d, Z_d) - H_n(X_d, Y_d, Z_d)- H_n(Z_d)$, as the sum of the volume related terms
\begin{align}
\sum_{i=1}^n \log(v_{XZ}(x_i,z_i)) + \sum_{i=1}^n \log(v_{YZ}(y_i,z_i)) \\
- \sum_{i=1}^n \log(v_{XYZ}(x_i,y_i,z_i)) - \sum_{i=1}^n \log(v_Z(z_i))
\end{align}
is equal to zero.
\end{proof}

To proof Theorem~\ref{th:convergence}, we link the above results:
\begin{equation}
\begin{split}
	&\lim_{v' \rightarrow 0} \lim_{n \rightarrow \infty} I^h(X;Y \mid Z) \\
= &\lim_{v' \rightarrow 0} \lim_{n \rightarrow \infty} (H^h(X,Z){+}H^h(Y,Z){-}H^h(X,Y,Z){-}H^h(Z)) \\
= &\lim_{v' \rightarrow 0} \lim_{n \rightarrow \infty} (H_n(X_d,Z_d){+}H_n(Y_d,Z_d){-} \\ 
& H_n(X_d,Y_d,Z_d){-}H_n(Z_d)) \\
= &\lim_{v' \rightarrow 0} (H(X_d, Z_d){+}H(Y_d, Z_d){-}H(X_d, Y_d, Z_d){-}H(Z_d))  \\
= &I(X;Y \mid Z) \; .
\end{split}
\end{equation}

\subsection{Implementation Details}
\label{apx:algorithm}

As discussed in Section~\ref{sec:algo}, our goal is to find that discretization that minimizes the joint entropy over a set of $k$ random variables via an iterative greedy algorithm.
We provide the pseudocode in Algorithm~\ref{alg:discretize:mixture}. As input, we are given a dataset $D = \{ D^1, \dots, D^k \}$ consisting of $n$ rows and $k$ columns, representing a sample of size $n$ from a $k$-dimensional random vector $X$, and a user-specified parameter $i_{\text{max}}$ specifying the maximum number of iterations. First, we initialize the discretization $X_d$ (line~\ref{a1:init}) by creating single bin histograms for the continuous points in $D_j$ and a bin with bin-width $1$ per discrete point. To detect the latter, we check if there exist $| \{ x \in \dX_j \mid D_j = x \} | \ge t$, where $t$ is a user-defined threshold. After that, we iteratively update the discretization for that $X_j$ providing the highest gain in stochastic complexity, until either the score cannot be improved or the maximum number of iterations has been reached (lines~\ref{a1:while}--\ref{a1:end:while}). To update the discretization of a variable $X_j$ we call the function $\text{refine}$ (line~\ref{a1:refine}), which receives as input the data $D_j$ and the discretization after iteration $i$. It then re-discretizes $X_j$ using an extension of the dynamic programming algorithm by Kontkanen et al.~\cite{kontkanen:07:histo}. In essence, instead of simply discretizing $X_j$ independently of the remaining variables, we keep the discretizations for all $X_i \neq X_j$ fixed and find the optimal histogram model $M^*$ over $X_j$ s.t. the overall score $L(D,M)$ is minimized.

\begin{algorithm2e}[tb!]
	\caption{}
	\label{alg:discretize:mixture}
	\Input{Data $D = \{ D_1, \dots, D_k \} \sim X$; \\
	Maximum number of iterations $i_{\text{max}}$}
	\Output{Discretization $X_d$}
	$X_d \leftarrow \text{init}(D)$; \label{a1:init}
	$i \leftarrow 1$\;
	\While {$X_d$ changes $\land \; i \le i_{\text{max}}$} { \label{a1:while}
		$X_d^i \leftarrow X_d$\;
		\ForEach{$j \in \{ 1, \dots, k \}$} {
			$X_d^{ij} \leftarrow \text{refine}(D_j \mid X_d)$\; \label{a1:refine}
			\If{$\Score(X_d^{ij}) < \Score(X_d^i$)}{
				$X_d^i \leftarrow X_d^{ij}$\;
			}
		}
		$X_d \leftarrow X_d^i$;
		$i \leftarrow i+1$\;
	} \label{a1:end:while}
	\Return{$X_d$}\;
\end{algorithm2e}

\subsection{Data Generation and Additional Experiments}
\label{apx:data-generation}

In the following, we first provide an empirical analysis on how the number of bins depends on the number of samples, then we give the details of the data generation for the experiments carried out on the synthetic causal network and last we provide additional experiments to evaluate $I_{\mathcal{X}^2}$ and $I_{\text{SC}}$.

\subsubsection{Sample Size and Number of Bins}

As discussed in Section~\ref{sec:histograms}, an important requirement to ensure consistency is that the number of bins grows as a sub-linear function w.r.t. the number of samples. We demonstrate that MDL-optimal histograms have this desirable property when learned on one-dimensional Gaussian distributions in Figure~\ref{fig:check-k}: the number of bins $K$ grows with $n$, but slower than $\sqrt{n}$. In addition, for multi-dimensional data, for which we can only approximate the histogram model that minimizes $L(D,M)$, we observe that if the number of dimensions increases, the average number of bins per dimension decreases if we keep $n$ fixed.
\begin{figure}[t]
	\begin{minipage}[t]{0.5\linewidth}
	\includegraphics[]{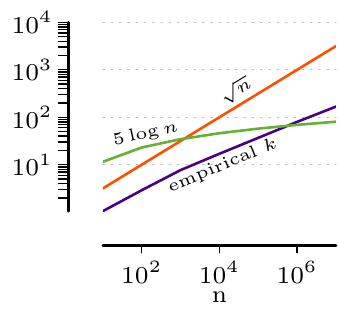}
	\end{minipage}%
	\begin{minipage}[t]{0.5\linewidth}
	\includegraphics[]{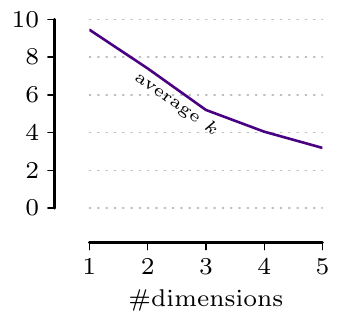}
	\end{minipage}%
	\caption{Left: Average number of bins $k$ to discretize $X \sim N(0,1)$ for increasing sample sizes ($20$ repetitions). Right: Per dimension of a multivariate Gaussian distribution with $X_i \Indep X_j$ and $X_i \sim N(0,1)$, we show the average number of bins ($n= 2\, 000$, $20$ repetitions).}
	\label{fig:check-k}
\end{figure}

\subsubsection{Synthetic Network}

Here, we describe the data generation for the synthetic network shown in Fig.~\ref{fig:synthetic-network}.
The source nodes of the network are $A$ and $B$. $A$ is generated as $A\sim \text{Exp}(1)$ and $B \sim \text{Unif}(0,4)$ (discrete). To get $B \to C$ we generate $C$ as $C \sim \text{Binom}(b,0.5)$ for $B=b$, for $B \to D$ we sample $D$ as $D \sim N(b-2,1)$ for $B=b$ and $E$ is sampled is exponentially distributed with rate $\frac{1}{c+1}$ for $C=c$. $F$ is generated as a function of $C$ and $D$. First, we generate $C'$ by rounding the values of $C$ and then we write $F$ as $F = D^{\frac{C'}{2}} + N(0,1)$. Last, we generate $G$ as the zero inflated Poissonization of $A$. Let $E' = \frac{\text{sign}(E-1)+1}{2}$, which ensures that $E'$ is either zero or one dependent on the value of $E$. Then $G \sim N(a,1)$ if $E'=0$ and $A=a$, and $G \sim \text{Poisson}(a)$ for $A=a$ if $G=1$.

\subsubsection{Detecting Collider and Non-Collider Structures}

\begin{figure}[t]
	\begin{minipage}[t]{0.5\linewidth}
	\includegraphics[]{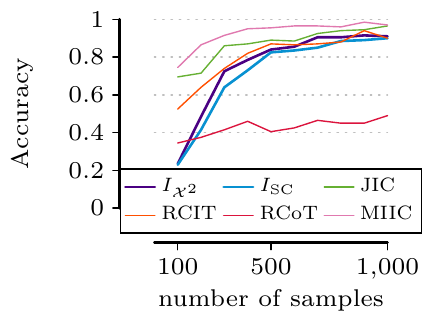}
	\end{minipage}%
	\begin{minipage}[t]{0.5\linewidth}
	\includegraphics[]{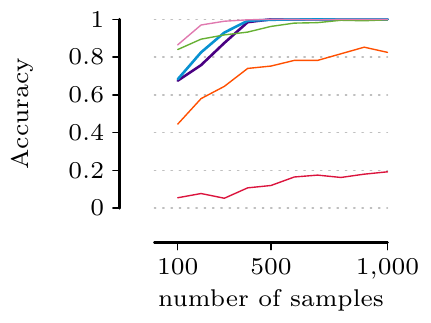}
	\end{minipage}%
	\linebreak
	\begin{minipage}[t]{0.5\linewidth}
	\includegraphics[]{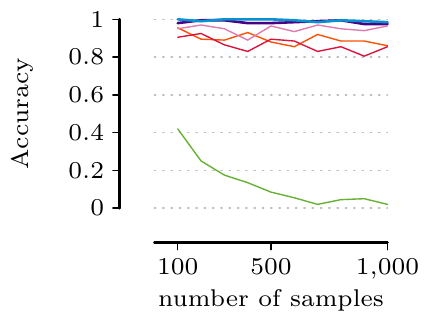}
	\end{minipage}%
	\begin{minipage}[t]{0.5\linewidth}
	\includegraphics[]{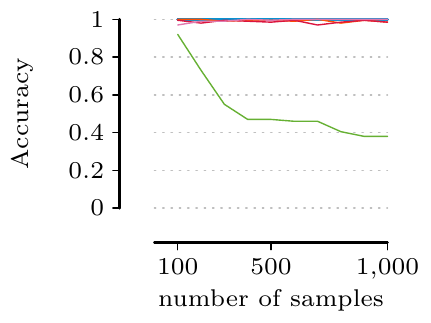}
	\end{minipage}%
	\caption{Accuracy for detecting continuous (left) and mixed-type (right) dependencies in collider structures (top) and independencies in non-collider structures (bottom) for different sample sizes.}
	\label{fig:collider}  
\end{figure}

To evaluate how well $I_{\mathcal{X}^2}$ and $I_{\text{SC}}$ can identify conditional (in)dependencies, we evaluate both variants on various generating mechanisms that involve collider and non-collider structures. Those structures are at the core of causal discovery, since collider structures can be inferred by detecting conditional dependencies, while non-collider structures impose conditional independencies. As in the causal discovery experiment, we set $\alpha = 0.01$ for $I_{\mathcal{X}^2}$, RCIT and RCoT. \\
\\
\textbf{Collider Structures} We generate data according to a collider structure, which can be represented by a directed acyclic graph as, e.g., $X \to Z \leftarrow Y$. According to this structure, we model $X$ and $Y$ by some distribution and write $Z$ as a non-deterministic function of $X$ and $Y$. We generate data for different generating mechanisms, including two continuous and four mixed settings.
\begin{enumerate}
	\item $X \Indep Y$ and $X,Y$ are either drawn from $N(0,1)$ or $\text{Uniform}( -2, 2 )$. $Z$ is an additive function of polynomials up to degree three or the tangent function plus additive noise $N\sim N(0,0.1)$---e.g. $Z = X^3 + \text{tan}(Y) + N$.  We pick the type of the distribution of $X, Y$, as well as the function type, uniform at random.
	\item $X,Y$ are drawn from a standard Gaussian distribution, with $X \Indep Y$ and $Z = \text{sign}(X \cdot Y) \cdot \text{Exp}(1 / \sqrt{2})$.
	\item $X,Y \sim N(0,1)$ with $X \Indep Y$ and $Z = \text{sign}(X \cdot Y)$, where we randomly assign a $z \in \dZ$ to $10 \%$ of the values in $Z$ to make the function non-deterministic.
	\item $X \sim N(0,1)$, $Y~\sim \text{Poisson}(\lambda)$, with parameter $\lambda$ selected uniformly at random from $\{ 1,2,3 \}$. We generate $Z$ as $X$ modulo $Y$ and assign $10 \%$ of the data points randomly.
	\item $X,Y$ are unbiased coins. $Z' = X \oplus Y$, where $\oplus$ denotes the xor operator. From $Z'$ we calculate $Z$ as $N(0,0.1)$ if $Z'=0$ and $\text{Poisson}(5) \cdot N(0,0.1)$ under the condition that $Z' = 1$.
	\item We generate $X,Y$ and $Z'$ as above, but this time we generate $Z$ as $\text{Poisson}(5) + N(0,0.1)$ if $Z'=1$ and as $N(0,0.1)$ if $Z' = 0$.
\end{enumerate}
For each generating mechanism, including two purely continuous and four mixed mechanisms, we generate $100$ data sets and report the averaged results, separately for the continuous and mixed data, in Fig.~\ref{fig:collider} (top). On the continuous data, both of our approaches perform on par with RCIT and JIC for more than $400$ data points, whereas MIIC has a slightly better performance and RCoT is not able to detect the dependence for the sign function and hence has an accuracy of about $50 \%$. Since the functions for mixed data include an xor and the modulo operator, it is difficult to treat all discrete variables as ordinal and hence RCIT only reaches up to $80 \%$ accuracy---which is mostly due to an xor determining the scaling of a Gaussian distributed variable. On the other hand, both of our tests perform very well and only need $400$ samples to obtain an accuracy close to $100\%$. JIC and MIIC perform on par with our tests. \\
\\
\textbf{Non-Collider Structures} Similar to collider structures, there also exist non-collider structures of the form $X \to Z \to Y$ or $X \leftarrow Z \to Y$. In both cases, the ground truth is that $X \Indep Y \mid Z$. To simulate data according to these graphs, we consider two continuous mechanisms based on polynomial functions and two mixed generating mechanisms.
\begin{enumerate}
	\item $X \sim N(0,1)$, $Z$ is an additive noise function of $X$ and $Y$ is an additive noise function of $Z$. The functions can be polynomials up to degree three or the tangent function.
	\item $Z \sim N(0,1)$, $X$ and $Y$ are independent additive noise functions of $Z$, as defined above.
	\item $X,Y$ and $Z$ are generated as in Experiment IV.
	\item $X$ and $Y$ are generated according to Experiment II and $Z \sim N(\mu, x)$ for $X=x$ and $\mu \in [ -4,4 ]$.
\end{enumerate}
In essence, Fig.~\ref{fig:collider} (bottom) shows that both our tests obtain almost perfect accuracies for the continuous and mixed data, whereas RCIT and RCoT fail to detect up to $20\%$ of the independencies for continuous data, MIIC does not detect up to $11 \%$ and JIC seems to generally overestimate dependencies for those test cases. If we consider these results in comparison to the results for detecting dependencies for the collider setting, we suspect that both MIIC and JIC have a larger tendency to falsely detect dependencies, while our approach is more conservative and hence needs more samples to detect true dependencies.
\fi

\end{document}